\documentclass{ecai} 

\usepackage{latexsym}
\usepackage{amssymb}
\usepackage{amsmath}
\usepackage{amsthm}
\usepackage{booktabs}
\usepackage{enumitem}
\usepackage{graphicx}
\usepackage{color}

\usepackage{xspace}
\usepackage{yfonts}
\usepackage{amsfonts}
\usepackage{stmaryrd}
\usepackage{hyperref}

\newtheorem{theorem}{Theorem}
\newtheorem{lemma}[theorem]{Lemma}
\newtheorem{corollary}[theorem]{Corollary}

\newcommand{\BibTeX}{B\kern-.05em{\sc i\kern-.025em b}\kern-.08em\TeX}

\begin{document}

\newcommand*{\nats}{\mathbb{N}}
\newcommand*{\ints}{\mathbb{Z}}
\newcommand*{\rats}{\mathbb{Q}}
\newcommand*{\reals}{\mathbb{R}}
\newcommand*{\bs}[1]{\boldsymbol{#1}}
\newcommand*{\stack}{\mathit{stack}}
\newcommand*{\hmax}{\mathit{hmax}}
\newcommand*{\relu}{\mathit{relu}}
\newcommand*{\scalar}[1]{\langle #1\rangle}
\newcommand{\countt}[2]{|#2|_{#1}}

\newcommand{\len}{\mathit{len}}

\newcommand*{\fixedfracrats}{\mathbb{Q}_f}
\newcommand*{\fpaset}{\mathbb{Q}^b_f}
\newcommand*{\satu}{\mathit{satu}}
\newcommand*{\wrap}{\mathit{wrap}}
\newcommand*{\rup}{\mathit{rup}}
\newcommand*{\rdown}{\mathit{rdown}}

\newcommand*{\downsatu}{F^\downarrow_\satu}
\newcommand*{\downwrap}{F^\downarrow_\wrap}
\newcommand*{\upwrap}{F^\uparrow_\wrap}
\newcommand{\upsatu}{F^\uparrow_\satu}

\newcommand*{\pwlfnn}{\mathcal{N}({\small\textit{PWL}})}
\newcommand*{\gad}[1]{\langle #1\rangle}
\newcommand*{\para}{|\!|}

\newcommand*{\emb}{\mathit{emb}}
\newcommand*{\pos}{\mathit{pos}}
\newcommand*{\att}{\mathit{att}}
\newcommand*{\score}{\mathit{score}}
\newcommand*{\norm}{\mathit{norm}}
\newcommand*{\pool}{\mathit{pool}}
\newcommand*{\comb}{\mathit{comb}}
\newcommand*{\out}{\mathit{out}}

\newcommand*{\alphC}{\textfrak{S}}
\newcommand*{\transC}{\mathcal{T}}
\newcommand*{\posC}{\mathcal{P}}
\newcommand*{\attC}{\mathcal{A}}
\newcommand*{\combC}{\mathcal{C}}
\newcommand*{\outC}{\mathcal{O}}

\newcommand*{\first}{\transC_1}
\newcommand*{\periodTransC}{\transC_\circ}

\newcommand{\ssmSAT}{\textsc{ssmSAT}\xspace}
\newcommand*{\SSM}{\ensuremath{\mathfrak{S}}\xspace}
\newcommand*{\tiSSM}{\ensuremath{\mathfrak{S}_{\text{ti}}}\xspace}
\newcommand*{\diagSSM}{\ensuremath{\mathfrak{S}_{\text{diag}}}\xspace}
\newcommand*{\gate}{\mathit{gate}}
\newcommand*{\inc}{\mathit{inc}}

\newcommand*{\empt}{\textsc{Empty}}
\newcommand*{\pcp}{\textsc{Pcp}}
\newcommand*{\unbOTP}{\ensuremath{\textsc{OTP}^*}\xspace}
\newcommand*{\expbOTP}{\ensuremath{\textsc{OTP}^{\mathsf{exp}}}\xspace}
\newcommand*{\unbOTWP}{\ensuremath{\textsc{OTWP}^*}\xspace}
\newcommand*{\expbOTWP}{\ensuremath{\textsc{OTWP}^{\mathsf{exp}}}\xspace}
\newcommand*{\expbSQTP}{\ensuremath{\textsc{SQTP}^{\mathsf{exp}}}\xspace}

\newcommand*{\MHALT}{\ensuremath{\textsc{MHALT}}\xspace}
\newcommand*{\Act}{\text{Act}}
\newcommand*{\minc}{\mathit{inc}}
\newcommand*{\mdec}{\mathit{dec}}
\newcommand*{\mztest}{\mathit{ztest}}

\newcommand{\size}{\mathit{size}}
\newcommand{\poly}{\mathit{poly}}
\newcommand{\NEXPTIME}{\textsc{NEXPTIME}\xspace}
\newcommand{\NP}{\textsc{NP}\xspace}
\newcommand{\PSPACE}{\textsc{PSPACE}\xspace}
\newcommand{\EXPSPACE}{\textsc{EXPSPACE}\xspace}
\newcommand{\RE}{\textsc{RE}\xspace}

\newcommand*{\LTL}{\textsc{LTL}\xspace}
\newcommand*{\LTLf}{\textsc{LTL}_f\xspace}
\newcommand*{\pLTLf}{\textsc{pLTL}_f\xspace}
\newcommand*{\cLTLf}{\textsc{pLTL}_f[$\backhash$]\xspace}
\newcommand*{\mLTLf}{\textsc{pLTL}_f[\MOD]\xspace}
\newcommand*{\uLTLf}{\textsc{un-pLTL}_f\xspace}
\newcommand{\yesterday}{\ensuremath{\mathop{\mathtt{Y}}}}
\newcommand{\ttrue}{\ensuremath{\mathop{\mathtt{t\!t}}}}
\newcommand{\ffalse}{\ensuremath{\mathop{\mathtt{f\!f}}}}
\newcommand{\since}{\ensuremath{\mathbin{\mathtt{S}}}}
\newcommand{\history}{\ensuremath{\mathop{\mathtt{H}}}}
\newcommand{\previously}{\ensuremath{\mathop{\mathtt{P}}}}
\newcommand{\hash}{\ensuremath{\textsc{\#}\xspace}}
\newcommand{\backhash}{\ensuremath{\mathop{\overleftarrow{\mathtt{\#}}\xspace}}}
\newcommand{\ltlnext}{\ensuremath{\mathop{\mathtt{X}}}}
\newcommand{\ltlfinally}{\ensuremath{\mathop{\mathtt{F}}}}
\newcommand{\until}{\ensuremath{\mathbin{\mathtt{U}}}}
\newcommand*{\nd}{\text{nd}\xspace}
\newcommand*{\Sub}{\text{Sub}\xspace}
\newcommand*{\MOD}{\mathtt{MOD}\xspace}


\newcommand{\TODOcomment}[2]{%
  \stepcounter{TODOcounter#1}%
  {\scriptsize\bf$^{(\arabic{TODOcounter#1})}$}%
  \marginpar[\fbox{
    \parbox{2cm}{\raggedleft
      \scriptsize$^{({\bf{\arabic{TODOcounter#1}{#1}}})}$%
      \scriptsize #2}}]%
  {\fbox{\parbox{2cm}{\raggedright
      \scriptsize$^{({\bf{\arabic{TODOcounter#1}{#1}}})}$%
      \scriptsize #2}}}
}%

\newcommand{\simpleTODOcomment}[2]{%
  \stepcounter{TODOcounter#1}%
  {\bf
    \scriptsize({\arabic{TODOcounter#1}~{#1}})
    {\bfseries{TODO:} #2}
  }
}

\newcounter{TODOcounter}
\newcommand{\TODO}[1]{\TODOcomment{}{#1}}
\newcommand{\TODOX}[1]{\simpleTODOcomment{}{#1}}

\makeatletter
\newcommand{\bigcomp}{%
  \DOTSB
  \mathop{\vphantom{\sum}\mathpalette\bigcomp@\relax}%
  \slimits@
}
\newcommand{\bigcomp@}[2]{%
  \begingroup\m@th
  \sbox\z@{$#1\sum$}%
  \setlength{\unitlength}{0.9\dimexpr\ht\z@+\dp\z@}%
  \vcenter{\hbox{%
    \begin{picture}(1,1)
    \bigcomp@linethickness{#1}
    \put(0.5,0.5){\circle{1}}
    \end{picture}%
  }}%
  \endgroup
}
\newcommand{\bigcomp@linethickness}[1]{%
  \linethickness{%
      \ifx#1\displaystyle 2\fontdimen8\textfont\else
      \ifx#1\textstyle 1.65\fontdimen8\textfont\else
      \ifx#1\scriptstyle 1.65\fontdimen8\scriptfont\else
      1.65\fontdimen8\scriptscriptfont\fi\fi\fi 3
  }%
}
\makeatother


\begin{frontmatter}


\paperid{8903} 


\title{The Computational Complexity of Satisfiability in \\State Space Models}


\author[A]{\fnms{Eric}~\snm{Alsmann}\thanks{Corresponding Author. Email: eric.alsmann@uni-kassel.de.}}
\author[A]{\fnms{Martin}~\snm{Lange}\thanks{Corresponding Author. Email: martin.lange@uni-kassel.de.}}

\address[A]{University of Kassel, Germany}


\begin{abstract}
    We analyse the complexity of the satisfiability problem ssmSAT for State Space Models (SSM), which asks whether an input sequence can lead the model to an accepting configuration. We find that ssmSAT is undecidable in general, reflecting the computational power of SSM. Motivated by practical settings, we identify two natural restrictions under which ssmSAT becomes decidable and establish corresponding complexity bounds. First, for SSM with bounded context length, ssmSAT is NP-complete when the input length is given in unary and in NEXPTIME (and PSPACE-hard) when the input length is given in binary. Second, for quantised SSM operating over fixed-width arithmetic, ssmSAT is PSPACE-complete resp.\ in EXPSPACE depending on the bit-width encoding. While these results hold for diagonal gated SSM we also establish complexity bounds for time-invariant SSM. Our results establish a first complexity landscape for formal reasoning in SSM and highlight fundamental limits and opportunities for the verification of SSM-based language models.
\end{abstract}

\end{frontmatter}

\section{Introduction}

State Space Models (SSM) \cite{guMambaLinearTimeSequence2024,guEfficientlyModelingLong2022,deGriffinMixingGated2024,10.5555/3692070.3694403,mehta2023long,sunRetentiveNetworkSuccessor2023,10.5555/3618408.3619518} have recently emerged as a promising alternative to transformer-based architectures 
for sequence modelling tasks.
By relying on linear recurrences and pointwise transformations, SSM are able to capture long-range dependencies 
while maintaining computational simplicity. This structural difference compared to transformers raises new 
questions about computational properties of SSM, in particular with regards to formal verification. 
While the empirical performance of SSM has been studied extensively, formal verification remains largely 
unexplored. In particular, understanding the complexity of basic reasoning tasks such as is there an input the SSM accepts (satisfiability) or does the SSM produce a specific output if the input satisfies some property (reachability) \cite{huangSurveySafetyTrustworthiness2024}, is essential for assessing the reliability and robustness of models deployed in safety-critical 
applications. Satisfiability serves as a baseline for investigating formal 
reasoning and formal verification due to its foundational role in capturing the core computational challenges 
inherent in these tasks. As detailed by Sälzer et al. \cite{salzerTransformerEncoderSatisfiability2024}, 
satisfiability abstracts a wide class of formal reasoning problems—including verification of safety properties 
and formal interpretability—by focusing on the existence of inputs that lead a model to exhibit a specific 
behavior. This abstraction removes dependency on domain-specific properties while preserving the essential 
computational characteristics. Take for instance a robustness property of an SSM like: given some set of suspicious 
symbols $E\subseteq \Sigma$, will every input containing some symbols of $E$ be rejected? This can be reduced to 
checking the satisfiability of the negated condition, i.e., whether there exists an input containing symbols 
from $E$ that is still accepted. A more general introduction to the different verification problems, in 
particular for language models, can be found in \cite{huangSurveySafetyTrustworthiness2024}. In this way, the 
satisfiability problem provides a unifying and theoretically robust framework for analysing the decidability 
and complexity of formal reasoning across a variety of verification and interpretability scenarios.

In this paper, we initiate a systematic investigation of the satisfiability problem for SSM. The satisfiability problem asks whether there exists an input sequence that causes the model to reach an accepting configuration. We first show that, in full generality, the problem is undecidable by a reduction from the halting problem for Minsky machines. Motivated by practical settings, we then consider two natural restrictions that make the problem decidable. First, we study SSM under bounded context length, a common assumption in deployed language models. Second, we investigate SSM operating over fixed-width arithmetic, reflecting the use of quantised computation in practical implementations. We investigate these problems for different practically motivated classes of SSM.

These results establish a first complexity landscape for formal reasoning in SSM. They provide a foundation for future work on verification techniques and a better theoretical understanding of the strengths and limitations of SSM.

\paragraph{Related Work.}

There is only limited work on formal properties of State Space Models. Recent theoretical work \cite{Terzic_Hersche_Camposampiero_Hofmann_Sebastian_Rahimi_2025, sarrofExpressiveCapacityState2024a, merrillIllusionStateStatespace2024} has exclusively dealt with the expressive power of SSM. Sarrof et al.~\cite{sarrofExpressiveCapacityState2024a} showed 
that SSM can express all star-free languages and that a 
specific class of SSM exactly characterises the star-free languages.
Furthermore, Merrill et al.~\cite{merrillIllusionStateStatespace2024} proved that quantised SSM can recognise only 
languages contained in the class $\text{TC}^0$. 

Complexity questions regarding the satisfiability problem have also been investigated for transformer architectures: Sälzer et al.~\cite{salzerTransformerEncoderSatisfiability2024} studied the satisfiability problem for 
transformer encoders, establishing complexity results similar to ours. 

\section{Preliminaries}

\paragraph{Notation.} We denote the set of natural numbers as $\nats$ and the set of real numbers as $\reals$. For a vector $\bs{v}$, we denote its $i$-th component as $v_i$. We use $\bs{0}$ to denote the zero vector, and $\bs{1}$ to denote the vector with all components set to $1$. We write $\bs{0}_d$ and $\bs{1}_d$ to denote the zero and one vectors of dimension $d$, respectively. For a finite set $M$, we fix an arbitrary but fixed ordering of its elements. We define the one-hot encoding $\mathbf{e}_M \colon 2^M \to \{0,1\}^{|M|}$ such that for any subset $S \subseteq M$, the vector $\mathbf{e}_M(S)$ has $1$ in each position corresponding to an element of $S$, and $0$ elsewhere. In particular, for any single element $a \in M$, we write $\mathbf{e}_M(a) := \mathbf{e}_M(\{a\})$. We denote the $i$-th standard basis vector by $\bs{e}_i$, with its dimension being clear from the context in which it is used. 
When convenient, we define large vectors by writing them as tuples of smaller vectors. For example, $(\bs{u}, \bs{v})$ denotes the vector obtained by concatenating the entries of $\bs{u}$ and $\bs{v}$.

\paragraph{Feedforward Neural Networks.} An \emph{(FNN-)node} is a function $v\colon\reals^k\to\reals$ with $v(\bs{x}) 
= \relu(\sum_{i=1}^{k} c_i x_i + b)$, where $k$ is the \emph{input dimension}, the $c_i \in \reals$ are called \emph{weights}, 
$b \in \reals$ is the \emph{bias} and $\relu: \reals \to \reals$ with $\relu(x) = \max(0,x)$ is the \emph{activation function} of $v$.
An \emph{(FNN-)layer} $l$ is a tuple of some $n$ nodes $(v_1, \dotsc, v_n)$ where each node has the same input dimension $m$. 
It computes the function $l\colon\reals^m\to\reals^n$ via $l(\bs{x}) = (v_1(\bs{x}), \dotsc, v_n(\bs{x}))$. We call 
$m$ the \emph{input} and $n$ the \emph{output dimension} of $l$.
A \emph{Feedforward Neural Network (FNN)} $N$ consists of $k$ layers $l_1, \dotsc, l_k$, where $l_1$ has input dimension $m$, the output dimension of $l_i$ is equal to the input 
dimension of $l_{i+1}$ for $i < k$ and the output dimension of $l_k$ is $n$. The FNN $N$ computes a function from $\reals^m$ to $\reals^n$ by $N(\bs{x}) = l_k(l_{k-1}( \ldots l_1(\bs{x}) \ldots ))$.

For our lower bound constructions we sometimes need to make use of FNN checking specific properties, for instance: is some 
vector entry equal to a specific value? For easier notation we will define FNN \textit{gadgets} for specific properties we 
will need later.
\begin{lemma}
    \label{lem:fnn_gadgets}
    Let $n,n_1, \cdots, n_k, m, b\in\nats$. There are FNN
    \begin{itemize}
        \item $N_{=b}$ s.t.\ $N_{=b}(n)= 1$ if $n=b$ and $N_{=b}(n)=0$ otherwise.
        \item $N_{\leq b}$ s.t.\ $N_{\leq b}(n)=1$ if $n\leq b$ and $N_{\leq b}(n)=0$ otherwise.
        \item $N_{\land}$ s.t.\ $N_{\land}(n_1, \cdots, n_k)=1$ if $n_i=1$ for all $i\leq k$ and $N_{\land}(n_1, \cdots, n_k)=0$ otherwise.
    \end{itemize}
\end{lemma}
\begin{proof}
    Let $N_{=b}$ be the FNN computing the function $N_{=b}(x)=\relu\left(\relu(x-(b-1))-2\cdot \relu(x-b)\right)$,
        $N_{\leq b}$ be the FNN computing $N_{\leq b}(x)=\relu\left(\relu(b+1-x)-\relu(b-x)\right)$ and
        $N_{\land}$ be the FNN computing $N_{\land}(x_1, \cdots, x_k)=N_{=k}(x_1 + \cdots + x_k)$.
        It is straightforward to see that the given FNN compute the required functions.
\end{proof}

We will sometimes need to combine FNN. We define two operations composition ($\circ$) and concatenation ($\|$). Let $N_1$ and $N_2$ be FNN computing functions $\reals^{m_i} \rightarrow \reals^{n_i}$ for $i\in\{1,2\}$. When $n_2=m_1$, the FNN $N_1 \circ N_2$ computes the function $\reals^{m_1} \rightarrow \reals^{n_2}$ with $N_1 \circ N_2(\bs{x})= N_1(N_2(\bs{x}))$. Syntactically, composition simply concatenates the output layer of $N_2$ and the input layer of $N_1$ setting all weights to $1$.

The FNN $N_1 \| N_2$ computes the function $\reals^{m_1+m_2} \rightarrow \reals^{n_1+n_2}$ with $N_1 \| N_2(\bs{x}_1, \bs{x_2}) = (N_1(\bs{x_1}), N_2(\bs{x_2}))$ for all $x_1\in\reals^{m_1}, x_2\in\reals^{m_2}$. Syntactically concatenation first extends $N_1$ to have the same number of layers as $N_2$ by adding identity nodes for additional layers. Then each layer of $N_1$ is extended by the respective layer of $N_2$ while zeroing incoming connections from nodes of $N_1$. 

\paragraph{State Space Models.}
We explore State Space Models as presented in \cite{guEfficientlyModelingLong2022}. For a structured analysis of this model, we will formalize its architecture following the approach in \cite{sarrofExpressiveCapacityState2024a}.

An SSM layer $l$ is defined as a tuple $(\bs{h}_0, \gate, \inc, \phi)$, where $\bs{h}_0 \in \reals^d$. The function $\gate$ is a mapping $\reals^d \rightarrow \reals^{d\times d}$, $\inc$ is a mapping $\reals^d \rightarrow \reals^d$, and $\phi$ is a mapping $\reals^d \times \reals^d \rightarrow \reals^d$. An SSM layer transforms an input sequence of vectors $\bs{x}_1 \cdots \bs{x}_k \in (\reals^d)^+$ into an output sequence $\bs{z}_1 \cdots \bs{z}_k\in (\reals^d)^+$ through the following process: it computes an intermediate sequence $\bs{h}_1\cdots\bs{h}_k \in (\reals^d)^+$ via a linear recurrence
$$
\bs{h}_t = \gate(\bs{x}_t) \cdot\bs{h}_{t-1} + \inc(\bs{x_t}) \quad \text{for } 1\leq t \leq k
$$
and subsequently generates the output via $\bs{z}_t = \phi(\bs{h}_t, \bs{x}_t)$.

An SSM comprising $L$ layers is expressed as a tuple $(\emb, l_1, \cdots, l_L, \out)$. Each $l_i$ refers to a layer 
as defined above, while $\emb$ is a function $\Sigma \rightarrow \reals^d$, and $\out$ is a function 
$\reals^d \rightarrow \reals^d$, computed by an FNN. The SSM computes a function $\Sigma^+ \rightarrow \reals$ as 
described: let $w=a_1 \cdots a_k\in\Sigma^+$ be a word. Initially, the SSM computes the embedding $\bs{x}_1^0\cdots \bs{x}_k^0$ 
of the word by setting $\bs{x}_i^0=\emb(a_i)$. Subsequently, for each layer $1\leq j\leq L$: compute 
$\bs{z}_1^j\cdots \bs{z}_k^j = l_j (\bs{x}_1^{j-1} \cdots \bs{x}_k^{j-1})$. Each layer’s output serves as the input for the
succeeding layer, meaning $\bs{x}_1^{j+1}\cdots \bs{x}_k^{j+1} = \bs{z}_1^{j}\cdots \bs{z}_k^{j}$. The SSM's final output, 
denoted as $\bs{y}_1\cdots \bs{y}_k$, is derived by applying $\out$ element-wise: $\bs{y}_i = \out(\bs{z}_i^L)$. In the end, 
the output of $\mathcal{S}(w)$ is computed by $\bs{y}_k$. We say that $\mathcal{S}$ accepts a word $w$ if $\mathcal{S}(w)=1$. 
Otherwise it is rejected. Let $\SSM$ be some class of SSM. We define the decision problem $\ssmSAT[\SSM]$ as the problem, given 
an SSM $\mathcal{S}\in\SSM$ over an alphabet $\Sigma$, decide whether there is $w\in\Sigma^*$ such that $\mathcal{S}(w)=1$.

To establish upper complexity bounds, we define the representation size of an SSM. Given an SSM over an 
alphabet $\Sigma$ with dimension $d$ and number of layers $L$, we set $|\mathcal{S}| = |\Sigma| + L + d$. We 
assume that the syntactic representation of $\mathcal{S}$ is polynomial in $|\mathcal{S}|$.
For any input word, the output of an SSM is computed layer-wise, with each layer requiring only a linear number 
of operations. This gives rise to the \emph{polynomial-evaluation property}: for any word $w \in \Sigma^*$, the 
value $\mathcal{S}(w)$ can be computed in time polynomial in $|\mathcal{S}| + |w|$. Note that we assume that each SSM can be finitely represented. This implies that all weights and constants used in the SSM have a finite representation.

\paragraph{Fixed-Width Arithemtic.}
Our notion of fixed-width arithmetics are representations of numbers using a fixed amount of bits, like floating- or fixed-point arithmetic. Our results are independent of the specific choice of an implementation. We assume that all values represented in a fixed-width arithmetic use $b$ bits for representing numbers. We say that an SSM works over fixed-width arithmetics if all computations and values occuring in the computation of the SSM are carried out using only $b$ bits.

\paragraph{Minsky-Machines.}
A Minsky machine $M=(Q, q_0, q_f, \delta)$ consists of a finite set of states $Q$, a start state $q_0\in Q$, a final state 
$q_f\in Q$ and a finite set of transitions $\delta \subseteq Q\times \text{Act} \times Q$ where 
$\text{Act}= \{\minc_1, \minc_2, \mdec_1, \mdec_2, \mztest_1, \mztest_2\}$. We additionally assume that for every state 
$q\in Q$ there is either only one $i\in\{1,2\}$  and one transition with action $\inc_i$ or two transitions with actions 
$\mdec_i$ and $\mztest_i$. A configuration of $M$ is a tuple $(q,c_1, c_2) \in Q\times \nats^2$, consisting of the current 
state and two counter values. A step of $M$ $(q, c_1, c_2) \vdash (q', c'_1, c'_2)$ is valid if (1) $(q, \inc_i, q')\in \delta$ 
and $c'_i = c_i+1$, (2) $(q, \mdec_i, q')\in \delta$, $c_i>0$ and $c'_i=c_i-1$ or (3) $(q, \mztest_i, q')\in \delta$, $c_i=0$ 
and $c'_i=c_i$ for $i\in\{1,2\}$. Note that the other counter in each step remains unchanged. A run of $M$ is a sequence valid steps starting from $(q_0, 0, 0)$. A run is accepting if it 
ends in a configuration with state $q_f$.

The problem to decide, given a Minsky machine $M$, whether it has an accepting 
run, is called $\MHALT$. Note that due to our definition of a Minsky machine $M$, given a specific state of $M$ there is always 
only one possible successor state. Thus, we can describe an accepting run of $M$ as a sequence of state-action pairs rather 
than full configurations.

\begin{theorem}[\cite{minskyComputationFiniteInfinite1967}] 
\label{thm:minskyundec}
    \MHALT is undecidable.
\end{theorem} 

\paragraph{Linear Temporal Logic $\LTLf$.} Let $\mathcal{P}$ be a finite set of atomic propositions. The syntax of $\LTLf$ is defined as follows:
\begin{align*}
    \varphi &::= p \mid \neg \varphi \mid \varphi \lor \varphi \mid \ltlnext \varphi \mid \varphi \until  \varphi
\end{align*}

Formulas of $\LTLf$
are interpreted over finite words over the alphabet $\Sigma = 2^{\mathcal{P}}$. Given a word $w=a_1\cdots a_n\in \Sigma^*$ and $i\in[n]$, the semantics of $\LTLf$ is inductively defined as follows:
\begin{alignat*}{2}
    w,i &\models p &&\iff p \in a_i \\
    w,i &\models \neg \varphi &&\iff w,i \not\models \varphi \\
    w,i &\models \varphi \land \psi &&\iff w,i \models \varphi \text{ and } w,i \models \psi\\
    w,i &\models \ltlnext\varphi &&\iff i< n \text{ and } w,i+1 \models \varphi\\
    w,i &\models \varphi\until\psi &&\iff  \text{ex. } i\leq k \leq n: w,k\models \psi \text{ and } \\
    & && \quad\qquad \text{f.a } \; i \leq j < k: w,j\models \varphi
 \end{alignat*}
We say that $w$ is a model of $\varphi$ iff $w,1\models \varphi$ (or simply $w\models \varphi$). For easier notation we define $\llbracket \varphi \rrbracket^w=1$ if $w \models \varphi$ and $0$ otherwise. We may omit the superscript if $w$ is clear from context. Furthermore, formula $\varphi$ is satisfiable iff there exists $w\in\Sigma^*$ such that 
$w$ is a model of $\varphi$. The size of $\varphi$ is denoted by $|\varphi|$. 

\begin{theorem}[\cite{degiacomoLinearTemporalLogic2013}]
    Satisfiability for $\LTLf$ is \PSPACE-complete.
\end{theorem}

\section{Overview}

\newtheorem{innercustomthm}{Theorem}
\newenvironment{customthm}[1]
{\renewcommand\theinnercustomthm{#1}\innercustomthm}
{\endinnercustomthm}

\begin{table*}[t]
    \caption{Overview of the results established in this paper.}
    \label{fig:results}
    \vspace{1em}
    \centering
        \setlength{\tabcolsep}{10pt}
\renewcommand{\arraystretch}{1.5}
        \begin{tabular}{lcc}
        \toprule
        & \textbf{time-invariant SSM} (\tiSSM) & \textbf{diagonal SSM} (\diagSSM) \\
        \midrule
        arbitrary precision & undecidable (Thm. \ref{thm:undec_ssm}) & undecidable (Thm. \ref{thm:undec_ssm})\\ 
        $\log$-precision & undecidable (Thm. \ref{thm:undec_ssm_log}) & undecidable (Thm. \ref{thm:undec_ssm_log})\\
        \midrule
        bounded context length (unary) & \NP-complete (Cor.~\ref{cor:bc_un_np_compl}) & \NP-complete (Cor.~\ref{cor:bc_un_np_compl})\\
        bounded context length (binary) & \NP-hard (Thm.~\ref{thm:bc_ssm}), $\in\NEXPTIME$ (Thm.~\ref{thm:bc_bin_nexptime}) & \PSPACE-hard (Thm.~\ref{thm:bc_bin_pspace_hard}), $\in\NEXPTIME$ (Thm.~\ref{thm:bc_bin_nexptime})\\
        \midrule
        constant bit-width & $\in \PSPACE$ (Thm.~\ref{thm:fix_const_mem_pspace}) & \PSPACE-complete (Cor.~\ref{cor:fix_diag_pspace_compl}) \\
        \midrule
        fixed bit-width (unary) & \NP-hard (Thm.~\ref{thm:fix_un_ti_np_hard}), $\in\PSPACE$ (Thm.~\ref{thm:fix_const_mem_pspace})& \PSPACE-complete (Cor.~\ref{cor:fix_diag_pspace_compl})\\
        fixed bit-width (binary) & \NP-hard (Thm.~\ref{thm:fix_un_ti_np_hard}), $\in\EXPSPACE$ (Thm.~\ref{thm:fix_bin_expspace})& \PSPACE-hard, $\in \EXPSPACE$ (Thm.~\ref{thm:fix_bin_expspace}) \\
        \bottomrule
        \end{tabular}
\end{table*}

Recent work on the expressiveness of SSM 
\cite{merrillIllusionStateStatespace2024,sarrofExpressiveCapacityState2024a} distinguishes between two main 
classes of SSM which differ in the allowed $\gate$ functions. The class of \textit{time-invariant} SSM
\cite{mehta2023long,sunRetentiveNetworkSuccessor2023,10.5555/3618408.3619518} (which we will denote \tiSSM) 
only allow $\gate$ functions such that there is a matrix $A\in\reals^{d\times d}$ with $\gate(\bs{x})=A$ 
for all $\bs{x}\in\reals^d$.
The class of \textit{diagonal} SSM \cite{guMambaLinearTimeSequence2024,deGriffinMixingGated2024,10.5555/3692070.3694403}
(which we denote \diagSSM) only use $\gate$ functions such that $\gate(\bs{x})$ is a diagonal matrix but can depend on $\bs{x}$. The output 
of each layer is computed by a non-linear function $\phi$ which also gets the initial input modelling a residual 
connection. In practice, different kinds of non-linear functions are used. In this paper we assume $\phi$ to be 
represented by an FNN. Despite this assumption, all upper complexity bounds established in this paper also hold 
for any other non-linear activation function used in practice.

An overview of the results established in this paper is depicted in Table~\ref{fig:results}. In Section \ref{sec:undecidability}, we establish the undecidability of the satisfiability problem for SSM (Thm.~\ref{thm:undec_ssm}). Our proof is based on a reduction from the halting problem for Minsky machines (\MHALT), and it proceeds by constructing, for each instance of a Minsky machine $M$, an SSM that accepts precisely those sequences that encode accepting runs of $M$. Finally, we show in Theorem~\ref{thm:undec_ssm_log} that this undecidability result persists even under finite-width arithmetic, provided the precision grows logarithmically with the input length.

In Section \ref{sec:decidability}, we turn our attention to decidable fragments of the problem. We consider two natural restrictions motivated by practical applications: bounding the input length and restricting arithmetic to fixed-width representations. In both cases we additionally make a distinction between unary and binary encodings of the word length or bit-width. While this distinction may seem subtle, it has significant implications for complexity. Intuitively, unary encoding corresponds to a setting where the available computational resources—such as the length of input sequences or the precision of arithmetic—are given explicitly and can be directly used during computation. In contrast, binary encoding specifies the size of these resources succinctly: we are told how much space we are allowed to use, but we must first construct or simulate that space ourselves.

For the bounded context length variant $\ssmSAT^{\leq}_{\{\text{un}, \text{bin}\}}[\SSM]$, we established that the unary-encoded version is \NP-complete (Corollary~\ref{cor:bc_un_np_compl}) for time-invariant and diagonal gated SSM. When considering binary encoding, we establish membership in \NEXPTIME\ (Theorem~\ref{thm:bc_bin_nexptime}) and \PSPACE-hardness (Theorem~\ref{thm:bc_bin_pspace_hard}) for diagonal SSM as well as \NP-hardness for time-invariant SSM (Theorem~\ref{thm:bc_ssm})

In the second setting, we study SSM operating over fixed-width arithmetic. We establish an exponential-model property (Lemma~\ref{lem:ssm_exp_model}), which ensures that if a word is accepted under $b$-bit precision, there exists an accepting word of length at most exponential in $|\mathcal{S}| + b$. We use this to prove that when the bit-width $b$ is constant or given in unary, the problem lies in \PSPACE\ (Theorem~\ref{thm:fix_const_mem_pspace}), while for binary-encoded bit-width the complexity increases to \EXPSPACE\ (Theorem~\ref{thm:fix_bin_expspace}).

To complement these upper bounds, we develop a polynomial-time reduction from the satisfiability problem for linear-time temporal logic over finite traces ($\LTLf$) to SSM satisfiability (Lemma~\ref{lem:ltl_ssm_reduction}). This construction relies on diagonal gated SSM with constant precision and shows that the corresponding fixed-width satisfiability problems are all \PSPACE-hard (Theorem~\ref{thm:fix_const_compl_pspace}). This yields \PSPACE-completeness for constant and unary bit-width (Corollary~\ref{cor:fix_diag_pspace_compl}), and \PSPACE-hardness for binary bit-width, leaving a gap to the \EXPSPACE\ upper bound.

For time-invariant SSM we establish \NP-hardness (Theorem~\ref{thm:fix_un_ti_np_hard}) for the cases with fixed-width arithmetic. Possible reasons for this complexity gap between time-invariant and diagonal gated SSM are discussed at the end of Section \ref{sec:decidability}. 
Our results provide a comprehensive classification of decidable fragments of the SSM satisfiability problem and reveal fine-grained complexity distinctions based on model structure and input encoding.

\section{Undecidability}
\label{sec:undecidability}

In this section, we will show undecidability for time-invariant as well as diagonal gated SSM working over arbitrary precision, by reducing from the halting problem for Minsky machines. Furthermore, it is shown that the undecidability still holds if we allow the precision to grow logarithmically with the input length. The reduction constructs an SSM which recognises accepting runs of a Misnky machine. In a first step, we will show how runs of a Minsky machine can be encoded as a word over a finite alphabet.

\begin{lemma}
    \label{lem:undec_l1}
    Let $M=(Q, q_0, q_f, \delta)$ be a Minsky-machine. There is an embedding function $\emb$ and an SSM layer $l_1\in \tiSSM \cap \diagSSM$ such that for each word $w=(q_1, a_1)\cdots(q_n, a_n) \in \Sigma^*$ with $\Sigma = \{(q', a) \mid \exists q\in Q : (q, a, q')\in \delta\}$ we have $l_1(\emb(w)) = \bs{z}_1\cdots\bs{z}_{|w|}$ with $\bs{z}_i = (\bs{e}_Q(q_i), \bs{e}_Q(q_i), \bs{e}_\Act(a), c^1_i, c^2_i, 0)$.
\end{lemma}

\begin{proof}

    Let $M=(Q, q_0, q_f, \delta)$ be a Minsky machine. We define $\Sigma=\{(q', a) \mid \exists q \in Q. (q, a, q') \in \delta\}$. Therefore, each word $w\in\Sigma^*$ describes a potential run of $M$, skipping the starting state. We now define An SSM $\mathcal{S}$ with one layer and $d=2|Q|+9$ dimensions which computes the counter states for the given run encoded in $w$. Let $\emb$ be the function $\Sigma \rightarrow \reals^d$ with:
    $$
        \emb((q, a)) = (\bs{e}_Q(q), \bs{e}_{Q}(q), \bs{e}_{\Act}(a), u_1(a), u_2(a), 0)
    $$
    where $u_i: \text{Act} \rightarrow \{-1, 0, 1\}$ with $u_i(\inc_i)=1$, $u_i(\mdec_i)=-1$, and $u_i(\mztest_i)=0$. Other inputs are mapped to zero.
    Let $l_{1} \in \tiSSM \cap \diagSSM$ computing 
    $$
        \bs{h}_t = E_{[2|Q|+7, 2|Q|+8]} \cdot \bs{h}_{t-1} + I_d \cdot \bs{x}_t \quad \phi(\bs{x}_t, \bs{h}_t)=\bs{h}_t
    $$
    where $E_{[i,j]}$ is the modified identity matrix with $E_{k,k}= 1$ if $i\leq k\leq j$ and $0$ otherwise. During this computation only the counter dimension get accumulated, resulting in 
    $$
        h_{t} = \left(\bs{e}_Q(q_t), \bs{e}_Q(q_t), \bs{e}_{\Act}(a), \sum_{1\leq i \leq t} u_1(a_i), \sum_{1\leq i \leq t} u_2(a_i), 0\right)
    $$
    This concludes that $l_1$ together with the embedding function, maps a potential run of $M$ in form of a state-action pair sequence to a sequence of vectors each containing the current state, the action and the intermediate counter values.
\end{proof}

Once we have encoded the current machine state and the corresponding counter values, we have to check whether the state-action pairs from the input are valid transitions with respect to $\delta$. Moreover, it has to be checked if the transitions taken were allowed with respect to the counter values.

To verify these properties, we must transfer information about the previous state into the vector at the next position. This can be achieved using a trick, which involves encoding the information about the previous state into the binary expansion of a value. This information can then be decoded back into a state using an FNN.

\begin{lemma}
    \label{lem:undec_prev}
    There is an SSM layer $l_{\text{prev}} \in \tiSSM \cap \diagSSM$ that maps a sequence $x_1 \cdots x_k \in \{0,1\}^*$ to a sequence $z_1 \cdots z_k \in \{0,1\}^*$ such that $z_i = x_{i-1}$ for $0 < i \leq k$ and $x_0 = 0$.
\end{lemma}
\begin{proof}[Proof Sketch.]
    To achieve this, we adopt an encoding technique used by \cite{sarrofExpressiveCapacityState2024a}, which encodes 
    the previous bit in the binary expansion of a value. This encoding allows the retrieval of historical information 
    using a simple FNN. This also works in a setting when finite-width arithmetic is used by ensuring that each bit of the history is separated by a zero in the 
    binary expansion, which avoids false results due to rounding errors and enables precise recovery via thresholding. The SSM layer computes the recurrence:
    $$
        h_t = \frac{1}{4}\cdot h_{t-1} + x_t
    $$
    This recurrence corresponds to a right shift of the binary expansion of $\bs{h}_t$ by two positions. This encodes the full history of $\bs{x}_1\cdots \bs{x}_{t-1}$ into the binary expansion of $\bs{h}_t$. Subsequently the previous bit can be read out by an FNN even when using fixed-width arithmetic with a bit-width of at least six. More details can be found in the Appendix.

\end{proof}

We now combine the previous lemmas to construct an SSM that verifies whether a given sequence of state-action pairs encodes an accepting run of a Minsky machine.

\begin{theorem}
    \label{thm:undec_ssm}
    $\ssmSAT[\tiSSM \cap \diagSSM]$ is undecidable.
\end{theorem}
\begin{proof}[Proof Sketch]
    We establish a reduction from \MHALT to $\ssmSAT[\tiSSM \cap \diagSSM]$. For each instance $M=(Q, q_0, q_f, \delta)$ of \MHALT we construct an SSM $\mathcal{S}_M \in \tiSSM \cap \diagSSM$, which accepts only valid encodings of accepting runs of $M$.

    $\mathcal{S}_M$ has $2|Q|+9$ dimensions and three layers. The embedding function $\emb$ and layer $l_1$ are given by Lemma \ref{lem:undec_l1} and are used to encode the state and the current counter values in each vector. Layer $l_2$, based on Lemma \ref{lem:undec_prev}, retrieves the previous state for each configuration. By setting $\bs{h}^2_0$ to encode the starting state, we ensure that the first state-action pair encoded a valid transition from the starting state. Additionally $l_2$ uses the output FNN to check the validity of each transition. Thus after $l_2$ the SSM computes the sequence $\bs{z}^2_1 \cdots \bs{z}^2_k$ with 
    $$\bs{z}^2_i=(\bs{e}_{Q}(q_i), \bs{e}_{Q}(q_{i-1}), \bs{e}_\Act(a), c^1_i, c^2_i, \text{check}(q_{t-1}, a, q_t, c^1_i, c^2_i)).$$
    $\text{check}(q_{t-1}, a, q_t, c^1_i, c^2_i)$ is the output of a pointwise applied FNN checking (1) if $(q_{t-1}, a, q_t) \in \delta$ and (2) if depending on the action, the counter states have a valid value. This verification is particularly crucial for $\mdec$ and $\mztest$ actions, where it must be ensured that the counter state is either not negative in case of $\mdec$ or is zero in case of $\mztest$. If the check is valid, the function will output zero and in case of an invalid transition it will be one.
    Consequently, layer $l_3$ applies a linear recurrence to accumulate values in the \textit{check} dimension. The output FNN then only has to check two things (1) whether \textit{check} has value zero, in which case each transition was valid or has a nonzero value, which leads to rejection of the input sequence and (2) if the last vector represents the accepting state $q_f$.
    This reduction establishes the undecidability of $\ssmSAT[\tiSSM \cap \diagSSM]$. The detailed proof can be found in the Appendix.
\end{proof}

Next, we consider the case when the SSM works over finite precision, but the precision depends logarithmically on the input length. This setting has recently been studied in terms of expressiveness \cite{merrillIllusionStateStatespace2024}. I.e.\ when evaluating an SSM on a word $w\in\Sigma^*$ with length $|w|$, the SSM works with $\mathcal{O}(\log(\max(|\Sigma|,|w|))$ precision.

\begin{theorem}
    \label{thm:undec_ssm_log}
    $\ssmSAT^{\text{fix}}_\text{log}[\tiSSM \cap \diagSSM]$ is undecidable.
\end{theorem}
\begin{proof}
    This proof follows the same lines as the proof of Theorem \ref{thm:undec_ssm}. Since we are working with binary representations, the number of bits required to represent a number that grows linearly with the input length increases only logarithmically. Therefore, a fixed-precision SSM that operates with $\mathcal{O}(\log(\max(|\Sigma|,|w|)))$ precision suffices to simulate the computation accurately. This ensures that all values required to verify a run can still be represented precisely enough to preserve correctness, which establishes that the undecidability result still holds under logarithmically bounded fixed-width arithmetic.
\end{proof}

\section{Decidable Cases}
\label{sec:decidability}

In this section, we investigate restrictions to the satisfiability problem for SSM that yield decidability.
A natural restriction, frequently encountered in practical language models, is bounded context length. Let  
$\SSM$ be a class of SSM. The bounded context length satisfiability problem for SSM, denoted by 
$\ssmSAT^{\leq}_{\{\text{un}, \text{bin}\}}[\SSM]$, is as follows: given $\mathcal{S} \in \SSM$ and $n \in \nats$ (encoded in unary or binary), decide whether there 
exists a word $w$ with $|w| \leq n$ such that $\mathcal{S}(w) = 1$. 

The second restriction, also frequently encountered in practice, is the use of SSM in quantised environments, where computations are performed using fixed-width arithmetic, such as fixed- or floating-point representations. We consider two different settings for SSM over fixed-width arithmetic. The first setting assumes the bit-width $b$ to be
constant. For $b \in \nats$, define the decision problem $b$-$\ssmSAT^{\text{fix}}[\SSM]$: given $\mathcal{S} \in \SSM$, does there exist a word $w$ such that $\mathcal{S}(w) = 1$ when $\mathcal{S}$ is evaluated over fixed-width arithmetic with bit-width $b$? The second setting considers the bit-width as part of the input. We distinguish between the cases where the bit-width is given in \emph{unary} or \emph{binary}. The decision problem $\ssmSAT^{\text{fix}}_{\{\text{un}, \text{bin}\}}[\SSM]$ is: given $\mathcal{S} \in \SSM$ and $b\in\nats$ (encoded in unary or binary), decide whether there exists a word $w$ such that $\mathcal{S}(w) = 1$ when $\mathcal{S}$ is evaluated over fixed-width arithmetic with bit-width $b$.

\subsection{Upper Bounds}
We start the complexity analysis by establishing upper bounds for the mentioned problems. All of our upper bounds will hold for time-invariant as well as diagonal gated SSM. Therefore we will use $\SSM$ to refer to the class of $\tiSSM \cup \diagSSM$.

\begin{theorem}
    \label{thm:bc_ssm_np}
    $\ssmSAT_{\text{un}}^{\leq}[\SSM]$ is in \NP.
\end{theorem}
\begin{proof}
    We first show membership in \NP\ using a standard guess-and-check approach. Given an SSM $\mathcal{S} \in \SSM$ 
    and $n \in \nats$ (encoded in unary), we nondeterministically guess a word $w \in \Sigma^*$ of length at most 
    $n$ and check whether $\mathcal{S}(w) = 1$. By the polynomial-evaluation property, this verification can be 
    performed in polynomial time.

\end{proof}

\begin{theorem}
    \label{thm:bc_bin_nexptime}
    $\ssmSAT_{\text{bin}}^{\leq}[\SSM]$ is in \NEXPTIME.
\end{theorem}
\begin{proof}
    Let $\mathcal{S}\in\SSM$ be an SSM and $n$ be the binary encoded word length. Hence, the maximum word length to check is exponential in $n$. This naturally gives us in $\NEXPTIME$ algorithm to decide $\ssmSAT_{\text{bin}}^{\leq}[\SSM]$ by just guessing a word $w$ with $|w|\leq 2^n$ and then checking if $\mathcal{S}(w)=1$. Using the polynomial-evaluation property $\mathcal{S}(w)$ can be computed in time polynomial in $|\mathcal{S}|+|w|= |\mathcal{S}|+2^n$.
\end{proof}

Turning our attention to the fixed-width arithmetic restriction, we first show that SSM operating over fixed-width arithmetic possess an \emph{exponential-model property}: if an SSM $\mathcal{S}$ accepts some word $w$ when using $b$-bit fixed-width arithmetic, then $\mathcal{S}$ also accepts a word $w'$ with $|w'| \leq 2^{\poly(|\mathcal{S}| + b)}$. Consequently, restricting SSM to fixed-width arithmetic has a similar effect to bounding the context length.

\begin{lemma}
    \label{lem:ssm_exp_model}
    Let $\mathcal{S} \in \SSM$ be an SSM with $L$ layers and $d$ dimensions, operating over fixed-point arithmetic with bit-width $b$. If $\mathcal{S}$ accepts a word $w$, then there exists a word $w'$ of length at most $2^{2 L d b}$ such that $\mathcal{S}$ accepts $w'$.
\end{lemma}
\begin{proof}
    Consider $\mathcal{S}$ with $L$ layers and $d$ dimensions, using $b$ bits of fixed-point precision. For each position $1 \leq i \leq n$ and each layer $0 \leq l \leq L$, $\mathcal{S}$ computes vectors $\bs{h}_i^l$. Each $\bs{h}_i^l$ depends only on $\bs{h}_{i-1}^l$ and the input $\bs{x}_i^{l-1}$, and since each entry is represented with $b$ bits, there are at most $2^{d b}$ possible distinct vectors per layer. Across all $L$ layers, the total number of possible combined layer states is at most $2^{L d b}$.

    Now, along the computation on a word of length $n > 2^{2 L d p}$, by the pigeonhole principle, there must exist indices $1 \leq i < j \leq n$ such that, for all $1 \leq l \leq L$, $\bs{h}_i^l = \bs{h}_j^l$ and $\bs{h}_i^{l-1} = \bs{h}_j^{l-1}$. That is, the complete vector state across all layers at positions $i$ and $j$ coincides. Therefore, the computation from position $j+1$ onward is identical to the computation from $i+1$ onward. Thus, the word obtained by removing the segment $a_{i+1} \cdots a_j$ from $w$ is also accepted by $\mathcal{S}$. Iterating this argument, we obtain an accepted word $w'$ of length at most $2^{2 L d p}$.
\end{proof}

We will use this exponential-model property in order to prove membership in \PSPACE for the cases where the bit-width is constant or part of the input and encoded in unary.

\begin{theorem}
    \label{thm:fix_const_mem_pspace}
    $b$-$\ssmSAT^{\text{fix}}[\SSM]$ and $\ssmSAT^{\text{fix}}_{\text{un}}[\SSM]$ are in \PSPACE.
\end{theorem}
\begin{proof}
    The computation of each next vector $\bs{h}_i^l$ in an SSM layer depends only on the previous state $\bs{h}_{i-1}^l$ and the embedding of the current symbol. By Lemma~\ref{lem:ssm_exp_model}, it suffices to consider words of length up to $2^{2Ldb}$, where $L$ is the number of layers, $d$ the dimension, and $b$ the bit-width.

    We describe a nondeterministic algorithm using polynomial space. According to Savitch's Theorem 
    \cite{SAVITCH1970177} this suffices to prove membership in \PSPACE. 
    Maintain a counter for the current word length (requiring $2Ldb$ bits), and nondeterministically generate a word $a_1 \cdots a_k \in \Sigma^*$ with $k \leq 2^{2Ldb}$, symbol by symbol. After each symbol $a_i$ is chosen, compute $\bs{h}_i^l$ from $\bs{h}_{i-1}^l$ and $\emb(a_i)$ for all $1 \leq l \leq L$, discarding $\bs{h}_{i-1}^l$ as it is no longer needed. Each step requires only polynomial time and space in $|\mathcal{S}|$. Storing all $\bs{h}_{i-1}^l$ for all layers requires $L d b$ bits. When $b$ is constant, the total space is $\mathcal{O}(|\mathcal{S}|^2)$. Thus, $b$-$\ssmSAT^{\text{fix}}[\SSM]$ is in \PSPACE. In the case of $b$ being part of the input and encoded in binary the total space is $\mathcal{O}((|\mathcal{S}|+b)^3)$. Thus, $\ssmSAT^{\text{fix}}_{\text{un}}[\SSM]$ is in \PSPACE.
\end{proof}

For binary-encoded bit-width, the problem lies in \EXPSPACE, due to the succinctness of the encoding.

\begin{theorem}
    \label{thm:fix_bin_expspace}
    $\ssmSAT^{\text{fix}}_{\text{bin}}[\SSM]$ is in \EXPSPACE.
\end{theorem}
\begin{proof}
    The proof is analogous to Theorem~\ref{thm:fix_const_mem_pspace}, except that with binary encoding of the bit-width, the space required to store an intermediate SSM computation is $L d 2^{|b|}$. This yields an \EXPSPACE\ algorithm for the problem.
\end{proof}

\subsection{Lower Bounds}
    Having established upper bounds for all of the problems, we now aim to show lower bounds. While the upper bounds did not depend on the class of gates the SSM used, we will see that this is not the case when it comes to lower bounds. 
\begin{theorem}
    \label{thm:bc_ssm}
    $\ssmSAT^{\leq}_{\text{un}}[\tiSSM \cap \diagSSM]$ is NP-hard.
\end{theorem}
\begin{proof}
    We reduce from the Zero-One Integer Programming Problem, which is known to be NP-complete \cite{karpReducibilityCombinatorialProblems1972}. Given an integer matrix $A \in \nats^{d\times d}$ and a vector $\bs{b} \in \nats^d$ (with numbers encoded in binary), the problem asks whether there exists $\bs{v} \in \{0,1\}^d$ such that $A\bs{v} = \bs{b}$.

    For a vector $\bs{v} \in \{0,1\}^d$, let $\mathrm{supp}(\bs{v}) = \{i \in \{1, \dots, d\} \mid v_i \neq 0\}$ denote 
    its support. Then $\bs{v}$ can be written as a sum of standard base vectors: 
    $\bs{v} = \sum_{i \in \mathrm{supp}(\bs{v})} \bs{e}_i$, so $A\bs{v} = \sum_{i \in \mathrm{supp}(\bs{v})} A\bs{e}_i$.

    We construct an SSM $\mathcal{S}_{A, \bs{b}}$ over the alphabet $\Sigma = \{1, \dots, d\}$, with dimension $2d$ 
    and embedding function $\emb(i) = (\bs{e}_i, \bs{0}_d)$. The input word encodes the support of $\bs{v}$ as a 
    sequence of indices which are mapped to the corresponding basis vectors. To ensure that each dimension appears 
    at most once, we duplicate the vector: the first $d$ components accumulate $A\bs{v}$, while the second $d$ 
    components count occurrences of each index. We define a single layer with recurrence $\bs{h}_0 = \bs{0}$ and
    $$
    \bs{h}_t = I \cdot \bs{h}_{t-1} + \begin{pmatrix}
    A & 0 \\
    I_d & 0
    \end{pmatrix} \bs{x}_t,
    $$
    so that
    $$
    \bs{h}_t = \left(\sum_{1 \leq i \leq t} A\bs{x}_i,\; \sum_{1 \leq i \leq t} \bs{x}_i\right).
    $$
    The output FNN $N_{\bs{b}}$ then verifies that the first $d$ entries equal $\bs{b}$ and each entry in the second 
    $d$ components is at most $1$, ensuring a valid encoding of $\bs{v} \in \{0,1\}^d$. Using the gadgets from 
    Lemma~\ref{lem:fnn_gadgets}, $N_{\bs{b}}$ can be defined as
    $$
        N_{\land} \circ \left( N_{=b_1} \| \cdots \| N_{=b_d} \| N_{\leq 1} \| \cdots \| N_{\leq 1} \right)
    $$

The resulting SSM $\mathcal{S}_{A, \bs{b}}$ has only polynomial size $A, \bs{b}$ yielding a polynomial time reduction. The constructed SSM is both time-invariant and diagonal. Thus, $\ssmSAT^{\leq}_{\text{un}}$ is NP-complete for both time-invariant and diagonal SSM.
\end{proof}

As a consequence of Theorem \ref{thm:bc_ssm} and Theorem \ref{thm:bc_ssm_np} we get the following result.
\begin{corollary}
    \label{cor:bc_un_np_compl}
    $\ssmSAT^{\leq}_{\text{un}}[\tiSSM]$, $\ssmSAT^{\leq}_{\text{un}}[\diagSSM]$ are NP-complete.
\end{corollary}

In order to prove \PSPACE lower bounds for some problems, the next lemma shows that we can reduce satisfiability of $\LTLf$ to \ssmSAT by constructing an SSM which is polynomial in the size of the formula and additionally only needs a constant bit-width to operate.

\begin{lemma}
    \label{lem:ltl_ssm_reduction}
    For any $\LTLf$-formula $\varphi$ over $\Sigma$, there exists an SSM $\mathcal{S}_\varphi \in\diagSSM$ such that for all $w=a_1\cdots a_n\in\Sigma^*$ holds: $w\models \varphi$ iff $\mathcal{S}_\varphi$ accepts $\overleftarrow{w}=a_n\cdots a_1$.
\end{lemma}
\begin{proof}
    
    Given a formula $\varphi$ over $\Sigma$, we construct an SSM $\mathcal{S}_\varphi$ over the same alphabet with $\mathcal{O}(|\varphi|)$ layers and dimensions, such that for all $w=a_1\cdots a_n\in\Sigma^*$, $w$ is a model of $\varphi$ if and only if $\mathcal{S}_\varphi$ accepts the word $\overleftarrow{w}=a_n\cdots a_1$. The reversal is due to the fact that while in $\LTLf$ one can specify properties about future positions, SSM can only transfer information about previous positions to the current state. By reversing the models the SSM can check a formula $\ltlnext \varphi$ by looking one step back and a formula $\varphi \until \psi$ by remembering that $\psi$ held in the past and $\varphi$ ever since.
    Moreover, $\mathcal{S}_\varphi$ requires only $6$ bits of precision, regardless of $\varphi$.

    We introduce the \emph{copy matrix} $C^{(i \rightarrow j)} \in \mathbb{R}^{d \times d}$, defined by $C^{(i \rightarrow j)} = \bs{e}_j \bs{e}_i^{\top}$. Multiplying $C^{(i \rightarrow j)}$ with a vector $\bs{x} \in \mathbb{R}^d$ yields a vector whose $j$-th entry is $x_i$ and all other entries are zero. The copy matrix is used to transfer values between specific dimensions in the SSM state vector, enabling the encoding of logical dependencies.

    Let $\varphi$ be an $\LTLf$-formula over the alphabet $\Sigma = 2^{\mathcal{P}}$, with $k$ subformulas and $k_{\ltlnext}$ $\ltlnext$-formulas. We construct 
    an SSM $\mathcal{S}_\varphi$ with $k+k_{\ltlnext}$ layers and $|\mathcal{P}| + k + 1$ dimensions. The embedding function 
    $\emb: \Sigma \to \mathbb{R}^{k+1}$ is defined as $\emb(a) = (\mathbf{e}_{\mathcal{P}}(a), \bs{0}_k, 1)$. 
    Since SSM layers are linear while the subformulas of $\varphi$ form a tree, 
    we fix a total order on $\varphi$'s subformulas $\varphi_1 < \cdots < \varphi_k = \varphi$ such that $i < j$ if 
    $\varphi_i$ is a subformula of $\varphi_j$. This order is used to construct the SSM layer-wise, ensuring that when 
    evaluating a subformula, all its dependencies have already been computed in lower layers. Each subformula, proposition 
    $a \in \mathcal{P}$, and the constant 1 correspond to a particular dimension in the vectors during the computation. 
    For easier notation we identify subformulas, propositions and the constant 1 with the corresponding dimension. 
    The layers $l_1, \ldots, l_k$, with $l_i \in \diagSSM$, are defined as follows:
    \begin{itemize}
        \item \emph{Case 1:} $\varphi_i = a$ for $a \in \mathcal{P}$. The value is read from dimension $a$ and copied to dimension $\varphi_i$:
        $$
        \bs{h}_t = 0 \cdot \bs{h}_{t-1} + (I + C^{(a \rightarrow \varphi_i)}) \cdot \bs{x}_t, \quad \phi(\bs{x}, \bs{h}) = \bs{h}.
        $$
        \item \emph{Case 2:} $\varphi_i = \neg \varphi_j$ with $j < i$. Compute $1$ minus the value at $\varphi_j$:
        $$
        \bs{h}_t = 0 \cdot \bs{h}_{t-1} + (I + C^{(1 \rightarrow \varphi_i)} - C^{(\varphi_j \rightarrow \varphi_i)}) \cdot \bs{x}_t, \quad \phi(\bs{x}, \bs{h}) = \bs{h}.
        $$
        \item \emph{Case 3:} $\varphi_i = \varphi_j \land \varphi_k$ with $j, k < i$. The conjunction is computed as $\max(0, \llbracket \varphi_j \rrbracket + \llbracket \varphi_k \rrbracket - 1)$:
        $$
        \bs{h}_t = 0 \cdot \bs{h}_{t-1} + (I + C^{(\varphi_j \rightarrow \varphi_i)} + C^{(\varphi_k \rightarrow \varphi_i)} - C^{(1 \rightarrow \varphi_i)}) \cdot \bs{x}_t,
        $$
        and $\phi(\bs{x}, \bs{h}) = \bs{h}'$, where $h'_j = \relu(h_j)$ if $j = \varphi_i$, and $h'_j = h_j$ otherwise.
        \item \emph{Case 4:} $\varphi_i = \ltlnext \varphi_j$ with $j < i$. In the reversed setting of the SSM, in order to evaluate $\ltlnext \varphi_j$ at some position we have to check the previous state vector. This is done in two consecutive layers. This first layer simply copies the evaluation of $\varphi_j$ to $\varphi_i$ at each position.
        $$
            \bs{h_t} = 0 \cdot \bs{h}_{t-1} + (I+C^{\varphi_j \rightarrow \varphi_i} \cdot \bs{x}_t) \quad \phi(\bs{x}, \bs{h})=\bs{h}
        $$
        The second layer uses Lemma~\ref{lem:undec_prev}, to construct a layer that overrides these values with the previous value. As stated in Lemma~\ref{lem:undec_prev} this requires only six bits of precision.
        \item \emph{Case 5:} $\varphi_i = \varphi_j \until \varphi_k$ with $j, k < i$. Unrolling the until operator we get:
        $$
        \varphi_j \until \varphi_k = (\varphi_j \land \ltlnext(\varphi_j \until \varphi_k)) \lor \varphi_k
        $$
        Because the SSM checks the formulas reversed, we check if since the last position where $\varphi_k$ was satisfied $\varphi_j$ held continuously. This check actually needs a non time-invariant diagonal gate, because the previous state of $\varphi_i$ has to be multiplied with the current state of $\varphi_j$ in order to check the reccurence. Additionally the current state of $\varphi_k$ is added to $\varphi_i$ to restart the recurrence as soon as $\varphi_k$ holds again. This can be encoded as
        $$
        \bs{h}_t = \mathrm{diag}(C^{(\varphi_j \rightarrow \varphi_i)} \bs{x}_t) \cdot \bs{h}_{t-1} + (I + C^{(\varphi_k \rightarrow \varphi_i)}) \cdot \bs{x}_t,
        $$
        and $\phi(\bs{x}, \bs{h}) = \bs{h}'$, where $h'_j = \min(1, h_j)$ if $j = \varphi_i$, and $h'_j = h_j$ otherwise. The function $min(1, x)$ is FNN-computable by $\relu(x)-\relu(-x)-\relu(x-1)$. The $\min$-operation is needed to scale down the resulting value to be either one or zero.
    \end{itemize}
    Let $w=a_1\cdots a_n\in\Sigma^*$ and $\bs{z}_1 \cdots \bs{z}_n$ the sequence of vectors after all layers defined above. Correctness of the reduction follows now from the fact that $w, i \models \varphi_j$ iff $(\bs{z}_{n-(i-1)})_{\varphi_j} = 1$.
    The output FNN simply has to check whether $\bs{z}_n$ at dimension $\varphi$ equals $1$.
    It remains to note that all matrix coefficients and intermediate values used in this construction can be represented using a constant number of bits, as well as intermediate numbers during the computation do not need more than six bits ensuring that the SSM operates even with constant bit-width. As stated above $\mathcal{S}_\varphi$ has polynomial size in $|\varphi|$.
\end{proof}

We note that in the previous lemma, evaluating the \textit{until}-operator required a time-dependent diagonal gate, due to the need to multiply the previous and next hidden state. Thus, the lower bounds based on the construction will only hold for the class \diagSSM. 
First, we look at the case of bounded context length SSM, where the length is given in binary. The previous lemma showed that we can find a polynomial time reduction from the satisfiability problem for $\LTLf$-formula to \ssmSAT. In order to make this reduction also work in the case of bounded context length we need an argument, that a satisfiable $\LTLf$-formula has a model of length at most exponential in the size of the formula.

\begin{lemma}[\cite{demriTemporalLogicsComputer2016,degiacomoLinearTemporalLogic2013}]
    \label{lem:ltl_small_model}
    For every satisfiable $\LTLf$ formula $\varphi$, there is a word $w$ with $|w|\leq |\varphi|\cdot 2^{|\varphi|}$ such that $w\models \varphi$.
\end{lemma}

Using this exponential-model property for $\LTLf$-formula we show a lower bound for diagonal gated SSM with bounded context length and with the length encoded in binary.

\begin{theorem}
    \label{thm:bc_bin_pspace_hard}
    $\ssmSAT^{\leq}_{\text{bin}}[\diagSSM]$ is \PSPACE-hard.
\end{theorem}
\begin{proof}
    We use the reduction from Lemma \ref{lem:ltl_ssm_reduction}. Given an $\LTLf$-formula $\varphi$, we can construct an SSM $\mathcal{S}_\varphi\in\diagSSM$ which accepts exactly the reversed models of $\varphi$. Lemma \ref{lem:ltl_small_model} showed that if $\varphi$ is satisfiable, it has a model $w\models \varphi$ with $|w|\leq |\varphi| \cdot 2^{|\varphi|}$. Due to Lemma~\ref{lem:ltl_ssm_reduction} the same holds for $\mathcal{S}_\varphi$. Therefore, we can specify the context length as $\log(|\varphi| \cdot 2^{|\varphi|})= \log(|\varphi|)+|\varphi|$ yielding a polynomial time reduction from the satisfiability problem for $\LTLf$-formula to $\ssmSAT^{\leq}_{\text{bin}}[\diagSSM]$.
\end{proof}

Next, we consider the case of SSM working over fixed-width arithmetic. First, we show a tight lower bound for the case of the bit-width being constant. 

\begin{theorem}
    \label{thm:fix_const_compl_pspace}
    $b$-$\ssmSAT^{\text{fix}}[\diagSSM]$ (for $b\geq 6$), $\ssmSAT^{\text{fix}}_{\text{un}}[\diagSSM]$ and $\ssmSAT^{\text{fix}}_{\text{bin}}[\diagSSM]$ are \PSPACE-hard.
\end{theorem}
\begin{proof}
    \PSPACE-hardness follows from the polynomial reduction from  $\LTLf$ to \ssmSAT in Lemma \ref{lem:ltl_ssm_reduction}. The resulting SSM only needed a constant bit-width of six, independently from the given formula. This of course also shows \PSPACE-hardness for $\ssmSAT^{\text{fix}}_{\text{un}}[\diagSSM]$ and $\ssmSAT^{\text{fix}}_{\text{bin}}[\diagSSM]$, because of the independence of the bit-width from the given formula.
\end{proof}

While this leaves an exponential complexity gap for $\ssmSAT^{\text{fix}}_{\text{bin}}[\diagSSM]$ as it is contained in $\EXPSPACE$ and $\PSPACE$-hard, it is a tight bound for the other two problems.

\begin{corollary}
    \label{cor:fix_diag_pspace_compl}
    $b$-$\ssmSAT^{\text{fix}}[\diagSSM]$ (for $b\geq 6$) and $\ssmSAT^{\text{fix}}_{\text{un}}[\diagSSM]$ are \PSPACE-complete.
\end{corollary}

Establishing an $\EXPSPACE$ lower bound for $\ssmSAT^{\text{fix}}_{\text{bin}}[\diagSSM]$ does not seem to be an easy task. Normally succinct encodings of problems imply an exponential increase in complexity. In this case though, only a portion of the input—the bit-width—is succinctly encoded, whereas the syntactical description of the SSM is given in standard (non-succinct) representation. It remains an open problem if this is enough to establish $\EXPSPACE$-hardness or if the upper bound can be lowered by exploiting some inherent properties of the SSM. We have seen that the construction in Lemma~\ref{lem:ltl_ssm_reduction} for establishing the \PSPACE-hardness depended on the use of non time-invariant diagonal gates, due to the recursive nature of the \textit{until}-operator. For time-invariant SSM we will show $\NP$-hardness as a lower bound.

\begin{theorem}
    \label{thm:fix_un_ti_np_hard}
    $\ssmSAT^{\text{fix}}_{\text{un}}[\tiSSM]$ and $\ssmSAT^{\text{fix}}_{\text{bin}}[\tiSSM]$  are \NP-hard.
\end{theorem}
\begin{proof}
    This follows by the same reduction as in Theorem~\ref{thm:bc_ssm} for the context length restriction. Given a matrix $A \in \nats^{d \times d}$ and vector $\bs{b} \in \nats^d$, we must show that the bit-width required for all intermediate computations of $A\bs{v}$ for some $\bs{v} \in \{0,1\}^d$ is polynomial in the size of $A$ and $\bs{b}$. Let $m$ be the largest entry in $A$. In the worst case, $\bs{v} = \bs{1}_d$, so the largest entry in $A\bs{v}$ is $d \cdot m$. Thus, all intermediate values can be represented with bit-width $\log(d) + \log(m)$, which is polynomial in the size of $A$ and $\bs{b}$. Thus, $\ssmSAT^{\text{fix}}_{\text{un}}[\tiSSM]$ is $\NP$-hard. The same obviously holds for $\ssmSAT^{\text{fix}}_{\text{bin}}[\tiSSM]$ as $\log(\log(d)+\log(m))$ is even smaller.
\end{proof}

This leaves the question why we could not close the complexity gap between time-invariant and diagonal SSM. Intuitively in time-invariant SSM, the calculation of a hidden state $\bs{h}_t=M\bs{h}_{t-1} + B(\bs{x}_t)$ is an affine linear transformation of $\bs{h}_{t-1}$, while in the case of diagonal SSM the transformation is non-linear, due to the multiplication of the input-dependent gate with the previous state. 

We have seen that this is necessary for our $\PSPACE$ lower bound. In the $\NP$ lower bound however, the hidden state computation only uses a linear transformation. It remains an open question if this non-linearity actually has an impact on the complexity and the \PSPACE, resp. \EXPSPACE upper bounds can be improved or if the lower bounds can be tightened despite the restriction to linear transformations.

\section{Outlook}

The satisfiability problem studied in this paper provides a natural foundation for the study on formal verification of State Space Models. Although abstract, it captures the core challenge underlying many verification tasks, such as proving or refuting safety properties. Its general formulation, detached from specific property types, ensures that undecidability and complexity results immediately extend to broader reasoning problems. Thus, satisfiability serves as a baseline for understanding the fundamental limits of verifying SSM and motivates future research into efficient and sound verification techniques. Establishing an $\EXPSPACE$ lower bound for $\ssmSAT^{\text{fix}}_{\text{bin}}[\diagSSM]$ remains an open problem. While the binary encoding of the bit-width introduces a form of succinctness, the SSM itself is not succinctly represented, leaving it unclear whether this suffices for $\EXPSPACE$-hardness or allows for a lower upper bound by exploiting structural properties of the model. 

Additionally, the complexity gap between time-invariant and diagonal SSMs has not been closed. A key distinction lies in the state update: time-invariant SSMs apply affine linear updates, whereas diagonal SSMs involve non-linear gating mechanisms. While such non-linearity is essential for our $\PSPACE$ lower bound, the $\NP$ lower bound relies solely on linear updates. It is an open question whether this non-linearity influences the overall complexity or whether tighter lower bounds can be achieved despite the linearity constraint. Future work should further investigate the role of non-linearity in complexity bounds and explore potential refinements of both upper and lower bounds for these SSM variants.

This paper demonstrates that sound and complete formal reasoning techniques for SSM face significant computational complexity challenges, even in settings with bounded context length and fixed-width arithmetic. Nevertheless, practical approaches such as SMT encodings of the verification problem or the development of incomplete or unsound verification procedures could offer promising pathways for ensuring desirable properties in SSM-based language models.

 

\bibliography{main}

\clearpage
\appendix
\section{Missing proofs from Section 4}

\begin{proof}[Proof of Lemma~\ref{lem:undec_prev}]
    
    Let $w\in\{0,1\}$ be some word $w=a_1 \cdots a_n$. For $0\leq k \leq n$ we define numbers $r_k\in\reals$ by $r_0=0$ and
$$
r_k = \frac{1}{4} r_{k-1} + a_k.
$$
We now show how an FNN can recover $a_{k-1}$ from $r_k$.

First observe that the binary encoding of $r_k$ has the following form:
    $$
    r_k = \texttt{$a_k$.0$a_{k-1}$0$a_{k-2}$0$\cdots$$a_{1}$}
    $$
    implying that one can distinguish four possible cases for $r_k$ depending on $a_k$ and $a_{k-1}$:
    \begin{enumerate}
        \item $a_k=0$, $a_{k-1}=0$: $(0.00)_2=0 \leq r_k\leq (0.00\overline{01})_2=1/12$
        \item $a_k=0$, $a_{k-1}=1$: $(0.01)_2 = 1/4 \leq r_k\leq (0.\overline{01})_2 = 1/3$
        \item $a_k=1$, $a_{k-1}=0$: $(1.00)_2 = 1 \leq r_k\leq (1.00\overline{01})_2 = 13/12$
        \item $a_k=1$, $a_{k-1}=1$: $(1.01)_2 = 5/4 \leq r_k\leq (1.\overline{01})_2 = 4/3$
    \end{enumerate}
    This approach remains valid under fixed-width arithmetic: by multiplying by $\frac{1}{4}$ instead of $\frac{1}{2}$, 
    the encoding shifts historical bits into every second binary position. This spacing mitigates rounding errors caused 
    by overflows, as each meaningful bit is buffered by a zero. For this reason there is no overlap in the intervals even 
    when numbers getting rounded up. This allows us to refine the intervals with borders which can be exactly represented 
    with only 6 bits avoiding any rounding problems. 

    \begin{enumerate}
        \item $a_k=0$, $a_{k-1}=0$: $(0.00)_2=0 \leq r_k\leq (0.001)_2=1/8$
        \item $a_k=0$, $a_{k-1}=1$: $(0.01)_2 = 1/4 \leq r_k\leq (0.1)_2 = 1/2$
        \item $a_k=1$, $a_{k-1}=0$: $(1.00)_2 = 1 \leq r_k\leq (1.001)_2 = 9/8$
        \item $a_k=1$, $a_{k-1}=1$: $(1.01)_2 = 5/4 \leq r_k\leq (1.1)_2 = 3/2$
    \end{enumerate}

    This mapping defines a piecewise-linear function which can be implemented exactly by an FNN. Thus, $l_\text{prev} \in \tiSSM \cap \diagSSM$ computes the following recurrence:
    $$
        h_t = \frac{1}{4} h_{t-1} + x_t
    $$
    The output function $\phi(x, h)$ then applies the corresponding piecewise-linear decision rule to extract the previous bit.
\end{proof}

\begin{proof}[Proof of Thm. \ref{thm:undec_ssm}]
    Let $M=(Q, q_0, q_f, \delta)$ be a Minsky machine. We define an SSM $\mathcal{S}_M$ with five dimensions and $d=2|Q|+9$ dimensions over the alphabet $\Sigma=\{(q', a) \mid \exists q \in Q. (q, a, q') \in \delta\}$. Let $\emb$ and $l_1$ be the embedding function and layer as defined in Lemma \ref{lem:undec_l1}. Hence, for a given state-action pair sequence $w=(q_1, a_1)\cdots (q_{n}, a_n)\in\Sigma^*$ we have $\bs{z}_1\cdots \bs{z}_n$ with
    $$
        \bs{z}_i = (\bs{e}_Q(q_i), \bs{e}_{|Q|}(q_i), \bs{e}_{\Act}(a_i) c^1_i, c^2_i)
    $$
    where $c^i_j$ are the accumulated counter states by adding up the effect of all actions up to this point. In order to check if the sequence of states actually respects the transition function of $M$, we need to gain access to the previous state. We achieve this by using the history bit construction from Lemma \ref{lem:undec_prev}. Thus the second layer $l_2\in \tiSSM \cap \diagSSM $ computes $\bs{h}^2_0 = (\bs{0}_d, \bs{e}_Q(q_0), 0, 0, 0)$ and:
    $$
        \bs{h}^2_t = \frac{1}{4}E_{[|Q|+1, 2|Q|]} \cdot \bs{h}^2_{t-1} + I_d \cdot \bs{x_t}
    $$
    By using a FNN $N_{prev}$ implementing the piece-wise linear decision rule from \ref{lem:undec_prev}, we can read out the last state gaining:
    $$
        N_{prev}(\bs{h}^2_t) = (\bs{e}_Q(q_t), \bs{e}_Q(q_{t-1}), \bs{e}_{\Act}(a_t) ,c_i^1, c_i^2, 0)
    $$
    Now at each position we have to check that $(q_{t-1}, a_t, q_t)\in\delta$. We can do this by constructing an FNN $N_{\text{trans}}$ computing a function $\reals^{2|Q|+6} \rightarrow \reals$ with $N_{\text{trans}}(\bs{e}_Q(q), \bs{e}_Q(q'), \bs{e}_{\Act}(a))=0$ if $(q, a, q')\in \delta$ and 1 otherwise for all $q, q' \in Q$ and $a\in\Act$. $N_{\text{trans}}$ can be constructed, due to the \textit{finite-sample expressivity} \cite{10.1145/3446776} which states that each finite set of input-output pairs can be exactly represented by an FNN. 
    The last thing we need to check is if the $\mdec_i$ and $\mztest_i$ actions were valid with respect to the counter states. We define a FNN $N_{\text{valid}}$ which checks for every position $t$ (1) if $a=\mdec_i$ then $c_t^i \geq 0$ and (2) if $a=\mztest_i$ then $c_t^i=0$. For this we define a gadget FNN $N_{\rightarrow}(x,y)$ such that for all $x,y\in\{0,1\}$ holds $N_{\rightarrow}(x,y)=0$ if $x\rightarrow y$ and 1 otherwise. $N_{\rightarrow}(x,y)$ is obviously computed by $1-\min(1, 1-x+y)$ (note that the behavior of $N_{\rightarrow}(x,y)$ on other inputs is not relevant for us). Hence by using the gadgets defined in Lemma \ref{lem:fnn_gadgets} we can define $$N_{\text{valid}}(a_t, c^1_t, c^2_t) = N_{\rightarrow}(N_{=\mdec_i}(a_t), N_{\geq 0}(c_t^i)) + \cdots $$ for all of the above cases.
    To gain the final output FNN $\phi^2$ for layer $l_2$, we do a composition of $N_{prev}$, $N_{\text{trans}}$ and $N_{\text{valid}}$ in such a way that $N_{prev}$ only manipulates dimensions $|Q|+1,\cdots, 2|Q|$, $N_{\text{trans}}$ takes dimensions $1, \cdots, 2|Q|+6$ as input and writes the output into the last dimension $d$ and $N_{\text{valid}}$ takes dimensions $2|Q|, \cdots, 2|Q|+6$ as input and adds the output also to the last dimension $d$. Therefore the final output after layer $l_2$ is $\bs{z}^2_t$ with:
    $$
        (\bs{e}_Q(q_t), \bs{e}_Q(q_{t-1}), \bs{e}_{\Act}(a_t) ,c_i^1, c_i^2, \delta(q_{t-1}, a_t, q_t)+\text{val}(a_t, c_i^1, c_i^2))
    $$
    Note that due to our construction the last dimension is greater than zero if something was not valid. Therefore, the last layer $l_3$ simply sums up the last dimension and the output FNN simply has to check, whether the sum of all last dimension at the last position is zero.

    This yields a reduction from \MHALT to $\ssmSAT[\tiSSM \cap \diagSSM]$ rendering the problem undecidable.

\end{proof}

\end{document}